\documentclass[11pt]{article}

\usepackage[T1]{fontenc}
\usepackage[margin=1in]{geometry}
\usepackage{amsmath,amsthm,amssymb}
\usepackage{mathptmx}
\DeclareMathAlphabet{\mathcal}{OMS}{cmsy}{m}{n}
\SetMathAlphabet{\mathcal}{bold}{OMS}{cmsy}{b}{n}
\usepackage{mathtools}
\usepackage{booktabs}
\usepackage{graphicx}
\usepackage{float}
\usepackage{placeins}
\usepackage{enumitem}
\usepackage{microtype}
\usepackage{setspace}
\usepackage{natbib}
\bibpunct[, ]{(}{)}{,}{a}{}{,}

\usepackage[dvipsnames]{xcolor}
\usepackage{xurl}
\usepackage{hyperref}
\usepackage[nameinlink,noabbrev]{cleveref}
\hypersetup{
  colorlinks=true,
  linkcolor=MidnightBlue,
  citecolor=MidnightBlue,
  urlcolor=MidnightBlue,
  pdfborder={0 0 0},
  pdfauthor={Varun Bansal, Sinan Ertemel, Rajnish Kumar, Anh Pham},
  pdftitle={Selective Verification of Priority Claims in School Choice}
}
\setlist{itemsep=2pt,topsep=4pt}

\theoremstyle{plain}
\newtheorem{theorem}{Theorem}[section]
\newtheorem{proposition}[theorem]{Proposition}
\newtheorem{lemma}[theorem]{Lemma}
\newtheorem{corollary}[theorem]{Corollary}
\theoremstyle{definition}
\newtheorem{definition}[theorem]{Definition}
\newtheorem{example}[theorem]{Example}

\newcommand{\I}{\mathcal I}
\newcommand{\Sset}{\mathcal S}
\newcommand{\emptyseto}{\varnothing}
\newcommand{\Ch}{\operatorname{Ch}}
\newcommand{\DA}{\operatorname{DA}}
\newcommand{\SVDA}{\operatorname{SVDA}}
\newcommand{\LSV}{\operatorname{LSV}}
\newcommand{\1}{\mathbf 1}

\title{Selective Verification of Priority Claims in School Choice}
\author{Varun Bansal\thanks{Economics and Planning Unit, Indian Statistical Institute, Delhi, India. E-mail: \href{mailto:varun23r@isid.ac.in}{\nolinkurl{varun23r@isid.ac.in}}.}
\and Sinan Ertemel\thanks{Corresponding author. Department of Economics, Istanbul Technical University, Ma\c{c}ka, Istanbul, T\"urkiye. E-mail: \href{mailto:ertemels@itu.edu.tr}{\nolinkurl{ertemels@itu.edu.tr}}.}
\and Rajnish Kumar\thanks{Economics Department, Queen's Business School, Queen's University Belfast, United Kingdom. E-mail: \href{mailto:rajnish.kumar@qub.ac.uk}{\nolinkurl{rajnish.kumar@qub.ac.uk}} (Kumar); \href{mailto:a.pham@qub.ac.uk}{\nolinkurl{a.pham@qub.ac.uk}} (Pham).}
\and Anh Pham\footnotemark[3]}
\date{}

\begin{document}
\singlespacing
\maketitle

\begin{abstract}
In school choice, eligibility claims determine priorities, but checking every claim before assignment is costly. We study selective verification in student-proposing deferred acceptance. Our main result characterizes exactly when an unresolved claim can affect a school's choice: the claimant must be able to cross the admission cutoff under some resolution of the other claims. This cutoff characterization yields an adaptive verification procedure that recovers the fully verified choice at each decision and preserves the entire deferred-acceptance path. When unresolved claims are valid, reported-priority order minimizes checks among procedures that inspect only currently pivotal claims at a fixed school and candidate set; it uses at most the school's capacity in checks. Under exact-tier verification, truthful preference and eligibility reports are weakly dominant. Simulations quantify verification savings and show how checking claims only for provisional final assignees can miss a false claim that changes the allocation.
\end{abstract}



\section{Introduction}

Matching mechanisms often assign priority using facts reported by applicants, such as residence, sibling status, or family income. Verifying these claims is costly, yet an unresolved claim can affect an applicant's assignment and the options available to others. This creates a mechanism-design problem at the intersection of information acquisition and strategic allocation: how can a designer decide which claims to verify while preserving the outcome of the announced matching rule? We study this question in student-proposing deferred acceptance (DA), where a verification decision at one school can change the set of applicants considered at later steps.

In England, investigations of false address claims have led to rejected applications and withdrawn school offers \citep{SchoolsWeek2015,BBC2026Southend}. In Ahmedabad, admissions reserved for low-income families were cancelled after income claims were found to be false \citep{BusinessToday2024RTE}; subsequent reporting questioned how the resulting vacancies would be filled \citep{GujaratSamachar2024RTE}. In Vienna, temporary address registrations around school enrollment have raised concerns about the perceived fairness of admissions \citep{Zednik2025,ORF2026Wien}. These cases show how eligibility claims shape admission decisions and how uncovering false claims can leave schools managing disrupted assignments and vacancies.

Schools rank applicants by ordered eligibility tiers, with a public baseline order breaking ties within each tier. A positive claim can be checked perfectly: a correct report receives its claimed tier, whereas an incorrect report receives baseline priority. Fully verifying every claim before DA defines the benchmark. We ask which claims can remain unresolved at each decision while preserving this benchmark throughout the matching process.

The paper's main theoretical contribution is a necessary-and-sufficient characterization of pivotality. A claim is pivotal when changing its validity can change the school's chosen set for some resolution of the other unresolved claims. With ordered tiers, this occurs exactly when the applicant can cross the admission cutoff. The characterization expresses this condition using counts of resolved and unresolved applicants who may rank above the claimant. It identifies precisely which information can affect the current choice, even though the clearinghouse does not yet know the final priorities.

The characterization yields an adaptive local verification procedure. A school checks a currently pivotal claim, incorporates the result, and recomputes pivotality. When no pivotal claim remains, its choice is the same under every remaining completion of the information. The procedure therefore recovers the full-verification choice. Pivotality identifies sensitivity of the current choice; it does not imply that every currently pivotal claim must eventually be checked, because resolving another claim may make it nonpivotal.

Applying the local procedure at every DA decision gives \emph{selective-verification deferred acceptance} (SVDA). For every report profile and pattern of claim validity, SVDA reproduces the applications, tentative holdings, and rejections of fully verified DA, and therefore its final matching. This result shows how the local characterization composes along the mechanism's sequence of choices. It also explains why verification confined to provisional final assignments need not suffice: a false claim by a student who is ultimately unmatched may already have changed an earlier rejection and the subsequent applications.

We also examine how the order of verification affects information costs. When every unresolved claim is valid, checking currently pivotal claims in reported-priority order (higher claimed tiers first, with baseline priority breaking ties) minimizes the number of checks among pivotal-only procedures at a fixed school and candidate set. Such an invocation requires no more checks than the school's capacity. These are local, conditional statements: the comparison allows adaptive selection of pivotal claims, but fixes the candidate set and assumes that its outstanding claims are valid. Exact implementation, by contrast, holds under every audit order and does not require truthful claims. An example shows that no fixed audit order minimizes checks at every validation profile.

Finally, we connect verification to reporting incentives. Under exact-tier verification, truthful eligibility reporting weakly improves a student's effective priority at every school. Combining this observation with DA's respect for priority improvements and its strategy-proofness under fixed priorities establishes that truthful reporting of preferences and eligibility is a weakly dominant strategy. On the truthful-reporting path, all positive claims are valid, so the local audit-order result and capacity bound apply. The incentive result does not require truth-telling to be the unique weakly dominant report.

The computational evidence illustrates the verification requirements of the mechanism in stylized markets. In the baseline with three school-specific tiers, reported-priority SVDA performs approximately one check per student, reducing verification by 76.3 percent relative to checking every submitted positive claim and by 36.4 percent relative to checking every positive claim when it first reaches a school. Further exercises vary market characteristics and claim validity, and examine failures of verification confined to provisional assignees. The simulations measure the amount of information acquired while preserving the fully verified DA path; their patterns are specific to the stated data-generating process.

Section~\ref{sec:model} defines the environment and the full-verification benchmark. Section~\ref{sec:local} characterizes pivotality and establishes the local verification and audit-order results. Section~\ref{sec:svda} establishes path equivalence in DA, and Section~\ref{sec:incentives} studies joint preference and eligibility reporting. Section~\ref{sec:simulation} presents computational evidence. Section~\ref{sec:discussion} discusses the scope of the analysis and concludes.

\subsection{Institutional motivation}\label{sec:institutions}

Ordered tiers capture priority structures used in several admissions systems. New York City's kindergarten, pre-K, and 3-K admissions distinguish groups using residence, sibling links, and current enrollment \citep{NYC2026Kindergarten,NYC2026PreK,NYC2026ThreeK}. English admission authorities publish ordered oversubscription criteria, commonly involving statutory status, siblings, residence, and feeder schools; the admissions framework also provides for verification and withdrawal of offers obtained through false information \citep{UKAdmissionsCode2021,UKAdmissionsCriteria,UKHowToApply}. The empirical literature on residence manipulation documents why treating the facts behind these priorities as automatically known can be consequential \citep{BjerreNielsenEtAl2023,Zednik2025}.

These settings motivate the primitives without being identical to the baseline model. We study a school-wide tier order and a validation rule under which an incorrect claim gives baseline priority. In practice, a failed claim may leave another entitlement intact, or priorities may differ across seats. Boston's former walk-zone policy is an example of the latter distinction \citep{AbdulkadirogluPathakRothSonmez2005Boston,DurKominersPathakSonmez2018,KominersSonmez2016}. The baseline isolates the information needed for ordered-tier choice; fallback entitlements and seat-specific rules are discussed as directions beyond its scope in Section~\ref{sec:discussion}.

\subsection{Related literature}\label{sec:literature}

The analysis builds on stable matching and school choice. \citet{GaleShapley1962} establish stability of deferred acceptance, and \citet{DubinsFreedman1981} and \citet{Roth1982} establish its incentive properties. \citet{AbdulkadirogluSonmez2003} formulate school choice as a mechanism-design problem, followed by the redesign of major assignment systems \citep{AbdulkadirogluPathakRoth2005NYC,AbdulkadirogluPathakRothSonmez2005Boston}. We take the fully verified DA rule as given. The question is which priority information must be acquired to implement its sequence of choices. The incentive argument uses DA's respect for priority improvements \citep{BalinskiSonmez1999,Kominers2024}.

A related literature studies the structure of priorities and school choice functions, including efficiency and fairness \citep{Ergin2002,KojimaManea2010}, affirmative action and controlled choice \citep{HafalirYenmezYildirim2013,EhlersHafalirYenmezYildirim2014,EcheniqueYenmez2015}, distributional constraints \citep{KamadaKojima2015}, and seat-specific or overlapping priority rights \citep{KominersSonmez2016,DurKominersPathakSonmez2018,AygunBo2021,SonmezYenmez2022,KitaharaOkumura2024}. Our uncertainty concerns the facts that determine an applicant's position within a specified priority system. The ordered-tier model makes it possible to characterize exactly when resolving one such fact can change a school's choice.

The analysis is also related to game-theoretic work on costly, limited, or imperfect verification in mechanism design \citep{BenPorathDekelLipman2014,MylovanovZapechelnyuk2017,PereyraSilva2023,BallKattwinkel2025}. \citet{FotakisTzamosZampetakis2016} study selective verification in truthful welfare approximation without transfers. We study verification within a strategic matching mechanism: applicants submit preference and eligibility reports, while the verification rule determines the priorities used by DA. Our objective is exact implementation of the fully verified rule, rather than an optimization between verification expenditure and allocation quality.

Information acquisition in matching has been studied from several perspectives. Work on students' acquisition of information concerns uncertainty about schools or preferences \citep{ChenHe2022,HakimovKublerPan2023,BjerreNielsenChrisander2022}. \citet{BampisEtAl2024} study query minimization for stable matching under one-sided uncertainty about preferences, and \citet{FengLiu2024} study the information needed to determine assignments in school choice. We instead resolve the validity of reports that generate priorities. These reports are strategically supplied, and their relevance depends on the candidate set reached during DA.

Finally, work on auditability asks whether participants can verify that a clearinghouse followed its announced rule \citep{GrigoryanMoller2024,GrigoryanMoller2024theory,HakimovRaghavan2026,Moller2026,HuZhang2026}. We study the clearinghouse's verification of applicant information. The game-theoretic contribution is to link the information that matters for a sequential allocation decision to the incentives generated by the verification rule: pivotality admits an exact cutoff characterization, and truthful preference and eligibility reporting is weakly dominant under the stipulated validation rule.

\section{Model and Verification Environment}\label{sec:model}

There is a finite set of students $\I=\{1,2,\ldots,n\}$ and a finite set of schools $\Sset=\{a,b,c,\ldots\}$.
Each school $s\in\Sset$ has capacity $q_s\geq 1$. Student $i\in\I$ has strict preferences $P_i$ over $\Sset\cup\{\emptyseto\}$, where $\emptyseto$ means being unmatched. Let $R_i$ be the corresponding weak preference relation. A school is acceptable if the student ranks it above $\emptyseto$.

Each school $s$ has a strict, public baseline priority order $\pi_s$. We write $i\mathrel{\pi_s}j$ when $i$ has higher baseline priority than $j$ at $s$. An exam score, a public lottery, or another observable rule may determine this order.

Schools find every student acceptable and prefer to fill every seat. Once priorities are fixed, school $s$ chooses the highest-priority $\min\{q_s,|A|\}$ students from any candidate set $A$.

\subsection{Priority-generating facts and submitted claims}

For each school $s$, let $\mathcal L_s=\{0,1,\ldots,K_s\}$, with $K_s\ge1$, be its eligibility tiers. We call the case $K_s=1$ for every $s$ the \emph{binary-tier model}. Tier $0$ gives no special priority, and larger numbers give higher priority: $K_s>K_s-1>\cdots>1>0$.

For each student--school pair $(i,s)$, let $\theta_{is}\in\mathcal L_s$ be student $i$'s \emph{true eligibility tier} at school $s$, and write $\theta_i=(\theta_{is})_{s\in\Sset}$ for the vector of her true tiers.

Student $i$ submits a preference report $\widehat P_i$ and a tier-report vector $c_i=(c_{is})_{s\in\Sset}$, where $c_{is}\in\mathcal L_s$.
A report $c_{is}=0$ makes no eligibility claim at school $s$. A report $c_{is}=r>0$ claims tier $r$. In practice, the student may report facts and the clearinghouse may convert them into this tier.

The true type of student $i$ is $t_i=(P_i,\theta_i)$, whereas her submitted report is $\widehat t_i=(\widehat P_i,c_i)$. The clearinghouse observes $\widehat t_i$, but not $\theta_i$.

\subsection{Verification technology}

Verification is perfect and reveals whether a claimed tier is correct. The result of a positive claim is $v_{is}(c,\theta)=\1\{c_{is}>0\text{ and }c_{is}=\theta_{is}\}$. A report $c_{is}=0$ requires no verification and generates no special priority.

The \emph{fully validated effective tier} is
\begin{equation}\label{eq:effective-tier}
 a_{is}(c,\theta)
 =c_{is}v_{is}(c,\theta)
 =c_{is}\1\{c_{is}>0\text{ and }c_{is}=\theta_{is}\}.
\end{equation}
Thus a valid claim receives the reported tier, while an invalid claim receives tier $0$. A failed claim does not fall to another positive tier. We refer to this rule as \emph{exact-tier verification}; demotion to tier $0$ is the stipulated consequence of an invalid claim.

The effective tier $a_{is}$ is an administrative state variable, not another report. It is the tier used after the claim is evaluated and may be unknown before verification.

\subsection{Validated priorities and school choice}

A higher effective tier ranks first, and baseline priority breaks ties within a tier. For a fixed profile $a$, define school $s$'s priority order $\succ_s^a$ by
\begin{equation}\label{eq:priority}
 i\succ_s^a j
 \quad\Longleftrightarrow\quad
 \begin{cases}
 a_{is}>a_{js},&\text{or}\\
 a_{is}=a_{js}\text{ and }i\mathrel{\pi_s}j.
 \end{cases}
\end{equation}
For a candidate set $A\subseteq\I$, let $\Ch_s^a(A)$ be the highest-priority $\min\{q_s,|A|\}$ students in $A$ under $\succ_s^a$.

\subsection{Full-verification benchmark}

A matching $\mu$ assigns each student to a school or to $\emptyseto$. No school receives more than its capacity, and $i\in\mu(s)$ if and only if $\mu(i)=s$.

Fix $a$ and $P$. A pair $(i,s)$ blocks $\mu$ if $i$ prefers $s$ to her assignment and either $s$ has an empty seat or $i$ outranks an assigned student under $\succ_s^a$. A matching is individually rational if every matched student prefers her assignment to being unmatched. It is \emph{stable} under $(P,\succ^a)$ if it is individually rational and has no blocking pair.

Full verification yields $a(c,\theta)$ and the benchmark $\mu^C(\widehat P,c;\theta)=\DA\bigl(\widehat P,\succ^{a(c,\theta)}\bigr)$, where $\DA$ denotes student-proposing deferred acceptance \citep{GaleShapley1962,Roth1982,AbdulkadirogluSonmez2003}. This benchmark represents the outcome the platform would obtain if it resolved every positive claim before beginning the assignment process.

Verification is perfect, always available when requested, and permanent. Each student--school check has unit cost, so the verification burden is the number of reports checked. Checks at different schools count separately, even for the same student. The model has no fixed verification budget. Students care only about their final assignments and bear no cost from reporting or verification.

\subsection{Illustrative examples}\label{sec:examples}

\begin{example}\label{ex:claimfact}
Suppose school $a$ has tiers $\{0,1,2\}$. If student $1$ has true tier $\theta_{1a}=2$ and reports $c_{1a}=2$, her effective tier is $a_{1a}=2$. If she reports $c_{1a}=1$, exact-tier verification rejects the claim and gives her tier $0$. A student with no eligibility also receives tier $0$ after making any positive claim.
\end{example}

\begin{example}\label{ex:one-seat}
One school $a$ has one seat. Both students prefer it to being unmatched, and baseline priority is $1\mathrel{\pi_a}2$.
School $a$ has tiers $\{0,1,2\}$, and student $2$ claims tier $c_{2a}=2$. A valid claim gives student $2$ the seat. An invalid claim gives both students tier $0$, so student $1$ receives the seat. The claim must therefore be resolved.
\end{example}

\begin{example}\label{ex:irrelevant}
There is one school $a$ with two seats and three students, with baseline priority $1\mathrel{\pi_a}2\mathrel{\pi_a}3$.
Student $1$ claims tier $1$, student $3$ claims tier $2$, and student $2$ makes no claim.

Student $1$ is selected whether her claim is valid or not. If it fails, she still has the highest baseline priority among tier-$0$ students. Student $3$'s claim matters: a valid claim gives the set $\{3,1\}$, while an invalid claim gives $\{1,2\}$. Full verification checks both claims, but checking only student $3$ determines the correct choice.
\end{example}

\begin{example}\label{ex:chain}
There are two one-seat schools $a$ and $b$ and three students. Preferences are
\[
 P_1:a\succ b\succ\emptyseto,
 \qquad
 P_2:a\succ b\succ\emptyseto,
 \qquad
 P_3:b\succ\emptyseto.
\]
At school $a$, baseline priority ranks student $1$ above student $2$. At school $b$, baseline priority is $2\mathrel{\pi_b}3\mathrel{\pi_b}1$.
Student $2$ reports a positive tier at $a$.

If the claim is invalid, student $1$ obtains $a$. Student $2$ then applies to $b$ and displaces student $3$. The outcome is $1\mapsto a$, $2\mapsto b$, and $3\mapsto\emptyseto$. If the false claim were accepted, student $2$ would obtain $a$, student $3$ would keep $b$, and student $1$ would be unmatched. A local verification error can therefore change rejections at other schools.
\end{example}

\section{Which Reports Are Decision-Relevant?}\label{sec:local}

Fix a school $s\in\Sset$ and the students $A\subseteq\I$ currently under consideration. Let $U_s(A)=\{i\in A:c_{is}>0\text{ and the report has not yet been verified}\}$ be the students whose positive claims remain unresolved, and let $R_s(A)=A\setminus U_s(A)$ be the resolved students. Resolved students include those who made no positive claim at $s$, claimants whose reports were verified earlier in the mechanism, and students whose reports are immediately verifiable from administrative records. For each $j\in R_s(A)$, let $\bar a_{js}\in\mathcal L_s$ be her known effective tier. It equals $0$ if she made no positive claim or her claim failed, and $c_{js}$ if her claim was verified.

When no confusion is possible, write $R$ and $U$. For each $i\in U$, write $r_i:=c_{is}>0$ for her claimed tier at $s$.

A \emph{completion} $x\in\{0,1\}^{U}$ assigns a possible verification result to every unresolved claim. It induces the effective tiers
\begin{equation}\label{eq:completion-tier}
 a_{js}(x)=
 \begin{cases}
 \bar a_{js},&j\in R,\\
 r_j,&j\in U\text{ and }x_j=1,\\
 0,&j\in U\text{ and }x_j=0.
 \end{cases}
\end{equation}
Let $\Ch_s(A;x)$ denote school $s$'s choice from $A$ under the priority order induced by \eqref{eq:completion-tier}.

\begin{definition}\label{def:pivotal}
Fix school $s$, candidate set $A$, and the current information state. An unresolved positive report of student $i\in U$ is \emph{pivotal for school $s$ at $A$} if there exists a completion $x_{-i}\in\{0,1\}^{U\setminus\{i\}}$ such that
\[
 \Ch_s(A;x_i=1,x_{-i})
 \neq
 \Ch_s(A;x_i=0,x_{-i}).
\]
\end{definition}

Pivotality is local: it depends on both the school and the set of applicants currently under consideration. It is not a probability statement and does not require a prior over claim validity. A claim is pivotal when its validity changes the current choice for at least one possible resolution of the other unresolved claims.

The following lemma shows that a claim is pivotal if and only if, for some completion of the other claims, its claimant is chosen when the claim is valid and not chosen when it is invalid.

\begin{lemma}\label{lem:single-report}
Fix school $s$, candidate set $A$, unresolved student $i\in U$, and a completion $x_{-i}$ of all other unresolved reports. Changing $x_i$ from $0$ to $1$ moves only student $i$ upward, from tier $0$ to tier $r_i>0$, and leaves the relative priority of every pair of other students unchanged. Hence $\Ch_s(A;1,x_{-i})\neq\Ch_s(A;0,x_{-i})$ if and only if $i$ is not selected when $x_i=0$ and is selected when $x_i=1$. If the chosen set changes, exactly one previously selected student is displaced by $i$.
\end{lemma}

\begin{proof}
Only $i$ changes position. If she is selected at tier $0$, moving her up does not change the chosen set. If she is still rejected at tier $r_i$, the set again does not change. A change therefore occurs exactly when $i$ crosses the cutoff. She then enters and one selected student leaves.
\end{proof}

Fix $i\in U$ with claimed tier $r_i>0$. Pivotality compares her position at tier $r_i$ with her position at tier $0$.

First count the resolved students who outrank $i$ even when her claim is valid:
\begin{equation}\label{eq:L-general}
 L_s^A(i)
 :=
 \bigl|\{j\in R:\bar a_{js}>r_i\text{ or }(\bar a_{js}=r_i\text{ and }j\mathrel{\pi_s}i)\}\bigr|.
\end{equation}
These students outrank $i$ in both states.

Next define the set of resolved \emph{swing students}
\begin{equation}\label{eq:Dset}
 \mathcal D_s^A(i)
 :=
 \{j\in R:
 j\text{ outranks }i\text{ when }i\text{ is at tier }0,
 \text{ but not when }i\text{ is at tier }r_i\}.
\end{equation}
Each member of $\mathcal D_s^A(i)$ outranks $i$ only when $i$'s claim is invalid.

For unresolved students, define the set of \emph{joint movers}
\begin{equation}\label{eq:Jset}
 J_s^A(i)
 :=
 \{j\in U\setminus\{i\}:i\mathrel{\pi_s}j\text{ and }r_j>r_i\}.
\end{equation}
A student $j\in J_s^A(i)$ has lower baseline priority than $i$ but claims a higher tier. If $j$'s claim fails, she ranks below $i$ in both states of $i$'s claim. If it succeeds, she ranks above $i$ in both states. Thus $j$ moves with, rather than against, the cutoff comparison for $i$.

Let $S_s^A(i):=U\setminus\bigl(\{i\}\cup J_s^A(i)\bigr)$.
Every student in $S_s^A(i)$ can be made a swing student. If $j\mathrel{\pi_s}i$, set $j$'s claim to invalid. Then $j$ beats an invalid $i$ by baseline priority but loses to a valid $i$. If $i\mathrel{\pi_s}j$, then $j\notin J_s^A(i)$ implies $r_j\le r_i$. Set $j$'s claim to valid. Student $j$ then beats an invalid $i$, while a valid $i$ beats $j$ by tier or, at the same tier, by baseline priority.

Accordingly, define $\Delta_s^A(i):=|\mathcal D_s^A(i)|+|S_s^A(i)|$.
This is the largest number of students who can be made to outrank $i$ only when her claim is invalid.

\begin{proposition}\label{prop:pivotality-general}
Fix school $s$, candidate set $A$, the current information state, and unresolved student $i\in U$. The report of $i$ is pivotal for school $s$ at $A$ if and only if
\begin{equation}\label{eq:pivotal-general}
 \Delta_s^A(i)\ge1,
 \qquad
 L_s^A(i)<q_s,
 \qquad
 q_s\le L_s^A(i)+\Delta_s^A(i)+|J_s^A(i)|.
\end{equation}
\end{proposition}

The conditions describe a cutoff crossing. The first requires someone to move across $i$ when her status changes. The second allows a valid $i$ to be selected. The third allows enough students to rank above an invalid $i$ to exclude her.

\begin{proof}
Write $L_i=L_s^A(i)$, $\Delta_i=\Delta_s^A(i)$, $J_i=J_s^A(i)$, and $S_i=S_s^A(i)$.

For sufficiency, assume the three inequalities in \eqref{eq:pivotal-general} and let $t:=\max\{0,q_s-L_i-\Delta_i\}$.
The third inequality gives $t\le |J_i|$. Since $L_i<q_s$ and $\Delta_i\ge1$, we also have $L_i+t<q_s$.

Set exactly $t$ claims in $J_i$ to valid and the rest to invalid. Configure every member of $S_i$ as a swing student as described above. Resolved tiers remain fixed.

When $i$'s claim is valid, exactly the $L_i$ always-ahead resolved students and the $t$ valid joint movers outrank her. Since $L_i+t<q_s$, she is selected.

Keep the same completion but make $i$'s claim invalid. Every member of $\mathcal D_s^A(i)$ and $S_i$ now joins the $L_i+t$ students already above her. Their number is $L_i+t+|\mathcal D_s^A(i)|+|S_i|=L_i+t+\Delta_i\ge q_s$, so $i$ is rejected. Lemma~\ref{lem:single-report} then implies that her claim is pivotal.

For necessity, suppose $i$ is pivotal. Under some common completion she is selected when valid and rejected when invalid. Let $h^1$ and $h^0$ count the students above her in these two states. Then $h^1<q_s\le h^0$.
Because all students counted by $L_i$ outrank a valid $i$, $h^1\ge L_i$ and hence $L_i<q_s$.

Only students who outrank an invalid $i$ but not a valid $i$ can raise $h^0-h^1$. They belong to $\mathcal D_s^A(i)$ if resolved and to $S_i$ if unresolved. Therefore $h^0-h^1\le |\mathcal D_s^A(i)|+|S_i|=\Delta_i$.
Since $h^0>h^1$, we obtain $\Delta_i\ge1$.

Finally, an invalid $i$ can be outranked by at most the $L_i$ always-ahead students, the $|\mathcal D_s^A(i)|$ resolved swing students, and the other $|U|-1$ unresolved students. Hence $h^0\le L_i+|\mathcal D_s^A(i)|+|U|-1=L_i+\Delta_i+|J_i|$.
Since $q_s\le h^0$, the third inequality follows.
\end{proof}

\begin{example}\label{ex:pivotality-tiered}
Return to Example~\ref{ex:irrelevant}. School $a$ has capacity two and baseline priority $1\mathrel{\pi_a}2\mathrel{\pi_a}3$.
Student $1$ has unresolved report $r_1=1$, student $3$ has unresolved report $r_3=2$, and student $2$ is resolved at tier $0$.

For student $1$, $L_a^A(1)=0$. Student $3$ has lower baseline priority but claims a higher tier, so $J_a^A(1)=\{3\}$. There are no swing students, and $\Delta_a^A(1)=0$. Student $1$'s claim is not pivotal.

For student $3$, $L_a^A(3)=0$. Student $2$ is a resolved swing student, and unresolved student $1$ can be made one by setting her claim to invalid, so $\Delta_a^A(3)=2$ and $J_a^A(3)=\emptyset$. The test becomes $\Delta_a^A(3)=2\ge1$ and $L_a^A(3)=0<q_a=2\le0+2+0$, so student $3$'s claim is pivotal.
\end{example}

The next lemma shows that once no unresolved claim is pivotal, the school's choice no longer depends on how the remaining claims are resolved.

\begin{lemma}\label{lem:invariant}
Fix school $s$, candidate set $A$, and a current information state. If no unresolved report is pivotal for school $s$ at $A$, then $\Ch_s(A;x)$ is the same for every completion $x\in\{0,1\}^{U}$.
\end{lemma}

\begin{proof}
Suppose completions $x$ and $y$ produce different choices. Move from one to the other by changing one coordinate at a time. At some step, changing one claim must change the choice. That claim is pivotal, a contradiction.
\end{proof}

For a fixed candidate set, resolving a claim never creates a new pivotal claim. In particular, resolving a pivotal claim strictly shrinks the set of pivotal claims.

\begin{lemma}\label{lem:persistence}
Fix school $s$, candidate set $A$, and unresolved student $i\in U$. If $i$'s report is nonpivotal at the current information state, resolving the report of any other student $j\in U\setminus\{i\}$ cannot make $i$'s report pivotal, regardless of whether $j$ is found valid or invalid.
\end{lemma}

\begin{proof}
If $i$ is nonpivotal, changing her status leaves the choice unchanged under every completion of the other claims. Resolving $j$ only removes possible completions. It cannot create one at which $i$ changes the choice.
\end{proof}

Pivotal claims need not form an interval in baseline priority. Consider a one-seat school $a$ with $1\mathrel{\pi_a}2\mathrel{\pi_a}3\mathrel{\pi_a}4$. Suppose student $1$ is already verified at tier $1$, while students $2,3,4$ have unresolved reports $r_2=2$, $r_3=1$, and $r_4=2$.
Student $2$'s claim is pivotal. Student $3$'s is not, because verified tier-$1$ student $1$ already outranks her. Student $4$'s claim is pivotal if student $2$'s claim fails. The pivotal set is therefore $\{2,4\}$, not an interval in baseline order.

In the binary-tier model, by contrast, the pivotal claims always form an interval. Let $T$ be the verified-valid claimants, $F$ the known tier-$0$ students, and $U$ the unresolved claimants. For $i\in U$, write $B_s^A(i)=\{j\in A:j\mathrel{\pi_s} i\}$.
Specializing Proposition~\ref{prop:pivotality-general} gives
\[
 i\text{ pivotal}
 \quad\Longleftrightarrow\quad
 |T\cap B_s^A(i)|<q_s
 < |T|+|U|+|F\cap B_s^A(i)|.
\]
Moving down the baseline order weakly increases both $|T\cap B_s^A(i)|$ and $|F\cap B_s^A(i)|$. The left-hand inequality therefore defines a prefix of unresolved claimants, while the right-hand inequality defines a suffix. Their intersection is an interval, possibly empty.

\subsection{Adaptive selective verification}

For each school $s$, fix a strict public audit order $\rho_s$ after claims are submitted. It may be exogenous or based on claimed tiers and baseline priority. The allocation results hold for every such order, which remains fixed during the local procedure.

\begin{definition}\label{def:local-verification}
Fix $s$, $A$, and $\rho_s$. While a pivotal claim remains, verify the highest such claimant under $\rho_s$. Fix a valid claim at its reported tier and an invalid claim at tier $0$, then recompute pivotality. Stop when no pivotal claim remains and return the choice shared by all remaining completions.
\end{definition}

Let $\LSV_s^{\rho_s}(A)$ denote the chosen set returned by this procedure.

\begin{proposition}\label{prop:local}
For every school $s$, candidate set $A$, information state, and public audit order $\rho_s$, local selective verification terminates after finitely many checks and returns the set that school $s$ would choose from $A$ if every positive report in $A$ were validated in advance.
\end{proposition}

\begin{proof}
Each check removes one student from $U$, so the procedure ends. At that point, Lemma~\ref{lem:invariant} implies that every remaining completion gives the same choice. The true validation profile is one such completion, so this choice is exactly the full-verification choice from $A$.
\end{proof}

The audit order may change which claims are checked and how many checks are made, but not the chosen set.

\begin{example}\label{ex:audit-order}
Consider one school $a$ with capacity $q_a=2$ and baseline priority $1\mathrel{\pi_a}2\mathrel{\pi_a}3\mathrel{\pi_a}4$.
Students $1,3,4$ claim tier $1$, while student $2$ claims tier $0$. All three positive claims are initially pivotal.

If all three claims are valid, checking students $1$ and $3$ first determines the choice after two checks, while starting with student $4$ requires three. If all three claims are invalid, checking students $3$ and $4$ first determines $\{1,2\}$ after two checks, while starting with student $1$ again requires three.

These two profiles do not by themselves eliminate every fixed order. Under Definition~\ref{def:local-verification}, the order $3\mathrel{\rho_a}1\mathrel{\rho_a}4$ uses two checks in both: when all claims are valid it checks students $3$ and $1$, after which student $4$ is nonpivotal; when all claims are invalid it checks student $3$, which makes student $1$ nonpivotal, and then student $4$.

Consider therefore a third profile, in which only student $3$'s claim is valid. Checking students $3$ and $4$ resolves the choice $\{1,3\}$ after two checks, since student $1$ is then selected whether or not her own claim holds. Under $3\mathrel{\rho_a}1\mathrel{\rho_a}4$, by contrast, student $1$ is checked second and found invalid, which makes student $4$ pivotal and forces a third check. This profile rules out every order that begins with student $1$ and the order $3\mathrel{\rho_a}1\mathrel{\rho_a}4$, while the all-valid profile rules out every order that begins with student $4$ as well as $3\mathrel{\rho_a}4\mathrel{\rho_a}1$.

Together, the all-valid profile and this third profile rule out all six fixed orders. No fixed order therefore minimizes checks for every possible validation profile. The procedure must adapt to what earlier checks reveal. Every order still produces the same full-verification choice.
\end{example}

\subsection{Choosing the audit order}\label{sec:reported-priority-audit}

Proposition~\ref{prop:local} shows that the audit order never affects which students a school chooses; it affects only how many claims are checked before the choice is settled. Example~\ref{ex:audit-order} shows that no fixed order minimizes this number for every validation profile, because the best next check depends on what earlier checks reveal. We therefore ask a narrower question: which order is best when every outstanding claim is valid? This is the relevant case on the truthful-reporting path, where every positive claim is valid (Section~\ref{sec:incentives}).

When claims are valid, the characterization in Proposition~\ref{prop:pivotality-general} suggests an answer. Once a claimant is verified at a high tier, she outranks every valid claimant at a lower tier, and every valid claimant at the same tier with lower baseline priority. Each such verification raises $L_s^A(j)$ for those lower claimants $j$ and can push them to the point where $L_s^A(j)\ge q_s$, making their claims nonpivotal without checking them. Checking the strongest claims first exploits this effect, as Example~\ref{ex:tier-first} below illustrates. For school $s$ and submitted tiers $c$, define the \emph{reported-priority order} $\rho_s^c$ by
\begin{equation}\label{eq:reported-priority-order}
 i\mathrel{\rho_s^c}j
 \quad\Longleftrightarrow\quad
 \begin{cases}
 c_{is}>c_{js}, &\text{or}\\
 c_{is}=c_{js}\text{ and } i\mathrel{\pi_s}j.
 \end{cases}
\end{equation}
Higher claimed tiers come first, and baseline priority breaks ties within a tier. For an unresolved claimant, $c_{is}=r_i$, so $\rho_s^c$ ranks claimants exactly as the school would rank them if all claims were valid. In the binary-tier model, $\rho_s^c$ reduces to baseline priority among claimants.

The order $\rho_s^c$ depends only on submitted reports and the public baseline order, so it can be announced once claims are filed and held fixed throughout, as Definition~\ref{def:local-verification} requires. Using it as the audit order therefore yields a special case of local selective verification.

\begin{definition}\label{def:reported-priority-audit}
The \emph{reported-priority audit rule} is local selective verification (Definition~\ref{def:local-verification}) with audit order $\rho_s=\rho_s^c$.
\end{definition}

Because the rule is a case of local selective verification, it inherits exactness from Proposition~\ref{prop:local}. The question is only whether it uses as few checks as possible. We compare it with every \emph{pivotal-only} procedure: any adaptive procedure that, at each step, verifies some currently pivotal claim of its choosing and stops once no pivotal claim remains. Such procedures may choose the next claim using earlier results, so this class contains local selective verification under every fixed audit order.

\begin{proposition}
\label{prop:reported-priority-optimal}
Fix a school $s$, a candidate set $A$, and a current information state. Suppose every unresolved positive report in $U$ is valid. Then the reported-priority audit rule uses the fewest verifications among all pivotal-only procedures.
\end{proposition}

\begin{proof}
Write $M_i:=L_s^A(i)+\Delta_s^A(i)+|J_s^A(i)|$ for the maximum number of students that can be placed above an invalid claimant $i$ across completions. By Proposition~\ref{prop:pivotality-general}, an unresolved report of $i$ is pivotal exactly when $\Delta_s^A(i)\ge1$, $L_s^A(i)<q_s$, and $q_s\le M_i$.

Let $h$ be the highest currently pivotal claimant according to $\rho_s^c$. Two observations drive the proof: validating a lower claimant does not change the pivotality of $h$, whereas validating $h$ can eliminate the pivotality of lower claimants.

First, consider any unresolved claimant $j$ with $h\mathrel{\rho_s^c}j$.
Then $j\notin J_s^A(h)$. Indeed, $h\mathrel{\rho_s^c}j$ implies either $r_h>r_j$, or $r_h=r_j$ and $h\mathrel{\pi_s}j$, whereas membership in $J_s^A(h)$ would require $j$ to lie below $h$ in baseline priority and to report a strictly higher tier. Thus $j\in S_s^A(h)$. If $j$ is verified and found valid, she becomes a resolved swing student for $h$: her positive validated tier puts her above an invalid $h$, while a valid $h$ remains above her either by the higher reported tier or, at an equal tier, by baseline priority. Consequently, the contribution of $j$ simply moves from $S_s^A(h)$ to $\mathcal D_s^A(h)$. Therefore
\begin{equation}\label{eq:lower-audit-invariance}
 L_h'=L_h,
 \qquad
 \Delta_h'=\Delta_h,
 \qquad
 |J_h'|=|J_h|.
\end{equation}
In particular, validating any lower reported-priority claimant leaves the pivotality of $h$ unchanged. Every unresolved claimant above $h$ in $\rho_s^c$ is currently nonpivotal and, by Lemma~\ref{lem:persistence}, remains nonpivotal; a pivotal-only procedure therefore verifies only claimants below $h$ until it verifies $h$. Hence $h$ stays pivotal, and every terminating pivotal-only procedure must eventually verify $h$.

Second, consider the effect of validating $h$ on a lower reported-priority unresolved claimant $j$. Because $h\mathrel{\rho_s^c}j$, a valid $h$ necessarily outranks a valid $j$. There are two cases. If $h\mathrel{\pi_s}j$, then before verification $h\in S_s^A(j)$; after verification $h$ becomes a resolved student who necessarily outranks a valid $j$. Hence $L_j'=L_j+1$, $\Delta_j'=\Delta_j-1$, and $|J_j'|=|J_j|$. If instead $j\mathrel{\pi_s}h$, then $h\mathrel{\rho_s^c}j$ requires $r_h>r_j$, so $h\in J_s^A(j)$ before verification. Once validated, $h$ again becomes a resolved student who necessarily outranks a valid $j$, and $L_j'=L_j+1$, $\Delta_j'=\Delta_j$, and $|J_j'|=|J_j|-1$.
Thus in both cases
\begin{equation}\label{eq:higher-audit-effect}
 L_j'=L_j+1,
 \qquad
 \Delta_j'\le\Delta_j,
 \qquad
 M_j'=M_j.
\end{equation}
Validating a higher reported-priority claimant can therefore make a lower claimant nonpivotal, but it cannot make a lower nonpivotal claimant pivotal.

We now use an exchange argument. Take a verification-minimal terminating pivotal-only sequence. Since $h$ must eventually be verified, consider the verification immediately preceding $h$ whenever $h$ is not first, and denote that claimant by $j$. By the first step, $j$ lies below $h$ in $\rho_s^c$. Since $h$ is pivotal at every information state preceding its verification, it is pivotal just before $j$ is verified, so we may reverse the pair and verify $h$ first.

If $j$ remains pivotal after $h$, verify $j$ next. Since both reports are valid, the information state after the two checks is exactly the same under the orders $j,h$ and $h,j$.

If instead $j$ becomes nonpivotal after $h$, omit the verification of $j$. We must check that leaving $j$ unresolved does not create additional verification requirements. Since $j$ was pivotal before $h$, its exclusion condition $q_s\le M_j$ held; by \eqref{eq:higher-audit-effect}, $M_j$ is unchanged when $h$ is validated. Hence $j$ can become nonpivotal only because either $L_j'\ge q_s$ or $\Delta_j'=0$.

If $L_j'\ge q_s$, then every unresolved claimant $k$ below $j$ in $\rho_s^c$ satisfies $L_k'\ge L_j'\ge q_s$. Indeed, every resolved student who necessarily outranks a valid $j$ also necessarily outranks a valid claimant with weakly lower reported priority. Thus every such $k$ is already nonpivotal. If instead $\Delta_j'=0$, then both $\mathcal D_s^A(j)$ and $S_s^A(j)$ are empty after the verification of $h$, so every other unresolved claimant belongs to $J_s^A(j)$ and hence reports a strictly higher tier than $j$. Thus there is no unresolved claimant below $j$ in $\rho_s^c$.

It remains to consider unresolved claimants above $j$. For any such claimant $k$, the unresolved valid report of the lower reported-priority student $j$ belongs to $S_s^A(k)$; fixing $j$ at its valid tier merely moves its contribution from $S_s^A(k)$ to $\mathcal D_s^A(k)$ and therefore leaves $L_k$, $\Delta_k$, and $|J_k|$ unchanged, exactly as in \eqref{eq:lower-audit-invariance}. Hence the pivotality of every claimant above $j$ is the same whether $j$ is verified valid or left unresolved, and the same comparison applies at every later information state, since subsequent checks affect both branches identically. We have therefore covered every remaining unresolved claimant: skipping $j$ creates no new pivotal report. Moreover, by Lemma~\ref{lem:persistence}, the reports that are already nonpivotal remain so as subsequent verifications are performed. The remainder of the original sequence can thus be followed without any additional check.

This adjacent interchange moves $h$ one position earlier without increasing the number of checks, and it may strictly reduce them. Repeating the interchange moves $h$ to the first position of a verification-minimal sequence. After verifying $h$, every remaining unresolved report is still valid, so the same argument applies recursively to the new information state. Induction on the number of unresolved reports yields the reported-priority audit rule.
\end{proof}

\begin{corollary}\label{cor:binary-top-down}
In the binary-tier model, all positive claimants report tier $1$. Hence $\rho_s^c$ coincides with baseline priority $\pi_s$, and Proposition~\ref{prop:reported-priority-optimal} reduces to top-to-bottom baseline-priority auditing.
\end{corollary}

The proposition is deliberately local and conditional. It fixes a school and candidate set, assumes that the outstanding positive reports are valid, and compares procedures that verify only currently pivotal reports. When reports may be invalid, no fixed order is best for every realized validation profile, as Example~\ref{ex:audit-order} shows even in the binary case. The comparison class is restricted to pivotal-only procedures. These restrictions do not affect exact implementation, which holds for every audit order.

\begin{example}\label{ex:tier-first}
Consider one school $a$ with one seat and baseline priority $1\mathrel{\pi_a}2\mathrel{\pi_a}3$. Student $1$ is resolved at tier $0$, while students $2$ and $3$ have unresolved reports $r_2=1$ and $r_3=2$. Suppose both reports are valid. Initially both reports are pivotal. A baseline-only top-to-bottom rule checks student $2$ first; student $3$ remains pivotal, so two verifications are required. The reported-priority rule instead checks student $3$ first because $r_3>r_2$. Once student $3$ is validated at tier $2$, she obtains the seat regardless of student $2$'s unresolved tier-$1$ report, so student $2$ becomes nonpivotal and one verification suffices. The example illustrates why the reported tier must be ordered before baseline priority in the general model.
\end{example}

The same ordering argument gives a simple bound on verification burden. Let $V_s^{\mathrm{RP}}(A)$ denote the number of reports checked during one invocation of reported-priority auditing at school $s$ with fixed candidate set $A$.

\begin{corollary}\label{cor:capacity-check-bound}
Fix school $s$, candidate set $A$, and a current information state. If every unresolved positive report is valid, then
\[
 V_s^{\mathrm{RP}}(A)\le \min\{q_s,|U|\}.
\]
\end{corollary}

\begin{proof}
Suppose the rule audits $i_1,\ldots,i_m$ in this order. At the $k$th audit, every previously verified claimant $i_1,\ldots,i_{k-1}$ is higher than $i_k$ in reported-priority order. To see this, suppose instead that some earlier audited claimant $i_j$ with $j<k$ were lower than $i_k$ under $\rho_s^c$. At the step at which $i_j$ was audited, $i_k$ was either pivotal or not. If $i_k$ was pivotal, the reported-priority rule would have audited $i_k$ rather than the lower claimant $i_j$. If $i_k$ was nonpivotal, then, since the candidate set $A$ is fixed throughout the procedure, Lemma~\ref{lem:persistence} implies that $i_k$ remains nonpivotal after every subsequent verification, contradicting the fact that $i_k$ is audited at step $k$. Both cases are impossible, so every previously audited claimant lies above $i_k$ under $\rho_s^c$. Because the earlier audited reports are valid, the corresponding students are resolved at their reported tiers and necessarily outrank a valid $i_k$. Hence $L_s^A(i_k)\ge k-1$. Pivotality of $i_k$ requires $L_s^A(i_k)<q_s$ by Proposition~\ref{prop:pivotality-general}, so $k\le q_s$. The bound by $|U|$ is immediate.
\end{proof}

The bound is tight. For example, take a school with capacity $q_s=q$, one tier-$0$ student with highest baseline priority, and $q$ lower-baseline students who all file the same positive binary claim. If all $q$ claims are valid, every claimant remains pivotal until all $q$ reports have been checked.

The all-valid hypothesis creates a useful link between auditing and incentives. The reported-priority rule need not be optimal for arbitrary false claims, but it is locally verification-minimal among pivotal-only procedures on a report profile generated by truthful tier reporting. More generally, a mechanism that makes truthful reporting a dominant-strategy recommendation supplies the validity condition under which Proposition~\ref{prop:reported-priority-optimal} and Corollary~\ref{cor:capacity-check-bound} apply. The next section embeds selective verification in deferred acceptance. Section~\ref{sec:incentives} then establishes exactly this incentive property.

\section{Implementing Selective Verification in Deferred Acceptance}\label{sec:svda}

We now place local selective verification inside student-proposing DA. Let $\rho=(\rho_s)_{s\in\Sset}$ be the public audit orders, fixed after the tier reports are submitted. School $s$ may use the reported-priority order $\rho_s^c$ from Subsection~\ref{sec:reported-priority-audit}.

\begin{definition}\label{def:svda}
Given a submitted preference profile $\widehat P$, a tier-report profile $c$, a true eligibility profile $\theta$, and public audit orders $\rho$ fixed for that submitted profile, \emph{Selective-Verification Deferred Acceptance} proceeds as follows.

\begin{enumerate}[label=\textup{(\roman*)}]
    \item In round $1$, every student applies to her most-preferred acceptable school under $\widehat P_i$, if one exists.

    \item In each round $t\ge1$, school $s$ forms the candidate set $A_s^t=H_s^{t-1}\cup N_s^t$, where $H_s^{t-1}$ is the set tentatively held by $s$ after the preceding round, $N_s^t$ is the set of new applicants to $s$ in round $t$, and $H_s^0=\emptyset$.

    \item School $s$ applies local selective verification to $A_s^t$ using audit order $\rho_s$ and every verification result previously obtained at $s$. It tentatively holds the resulting set $H_s^t$ and rejects every student in $A_s^t\setminus H_s^t$.

    \item In the next round, every rejected student applies to her most-preferred acceptable school under $\widehat P_i$ that has not previously rejected her, if one exists.

    \item The procedure terminates when no student makes a new application. The tentative holdings at termination constitute the final matching.
\end{enumerate}
\end{definition}

Verification results are permanent. A nonpivotal claim is postponed, not ignored forever. If later applicants change the candidate set, the school recomputes pivotality before its next choice.

Denote the resulting matching by $\SVDA_{\rho}(\widehat P,c;\theta)$. The semicolon separates the submitted reports from the true tiers $\theta$, which determine the results of any checks.

\subsection{Outcome and path equivalence}

\begin{lemma}\label{lem:local-implementation}
Fix school $s$, a current candidate set $A$, submitted tier reports $c$, and true eligibility $\theta$. Suppose all verification results obtained before the current school choice are correct. Then local selective verification applied to $A$ returns $\Ch_s^{a(c,\theta)}(A)$, the choice that school $s$ would make from $A$ under the fully validated ordered-tier priority profile.
\end{lemma}

\begin{proof}
Proposition~\ref{prop:local} says that the local procedure returns the choice obtained by checking every claim in $A$. Earlier results are correct by assumption, and $c$ and $\theta$ determine the unresolved results. The choice is therefore $\Ch_s^{a(c,\theta)}(A)$.
\end{proof}

\begin{corollary}\label{cor:svda-equivalence}
Fix a submitted preference profile $\widehat P$, tier reports $c$, true eligibility $\theta$, and public audit orders $\rho$. Then
\begin{equation}\label{eq:svda-equivalence}
 \SVDA_{\rho}(\widehat P,c;\theta)
 =\mu^C(\widehat P,c;\theta)
 =\DA\bigl(\widehat P,\succ^{a(c,\theta)}\bigr).
\end{equation}
Moreover, SVDA and full-validation student-proposing DA generate the same applications, tentative holdings, and rejections in every round. In particular, the final matching is independent of the public audit orders $\rho$.
\end{corollary}

\begin{proof}
By Lemma~\ref{lem:local-implementation}, each school makes the full-verification choice from its current candidate set. The two procedures therefore have the same holdings and rejections in round $1$. If their histories agree through round $t$, the same students apply next, every school faces the same candidates in round $t+1$, and the choices again agree. Induction gives the same applications, holdings, rejections, and final matching. Local choices do not depend on the audit order, so neither does the outcome.
\end{proof}

When only the final matching matters, we write $\SVDA(\widehat P,c;\theta)$ and omit the audit order.

\begin{corollary}\label{cor:stability}
If preferences are reported truthfully, $\widehat P=P$, then for every $c$ and $\theta$ the SVDA outcome is stable under $(P,\succ^{a(c,\theta)})$ and is student-optimal among stable matchings under that priority profile.
\end{corollary}

\begin{proof}
Under truthful preferences, Corollary~\ref{cor:svda-equivalence} identifies SVDA with student-proposing DA under $\succ^{a(c,\theta)}$. The standard properties of DA give the result.
\end{proof}

\begin{corollary}\label{cor:pref-strategyproof}
Fix $c$ and $\theta$. Truthful preference reporting is a weakly dominant strategy under SVDA.
\end{corollary}

\begin{proof}
For every $\widehat P$, SVDA equals student-proposing DA under the fixed priorities $\succ^{a(c,\theta)}$. It therefore inherits student-side strategy-proofness.
\end{proof}

Checking only provisional assignees need not recover the full-verification matching. This procedure honors every claim provisionally, runs DA, and checks only claims made by provisionally assigned students. This may use fewer checks, but it can produce the wrong matching because a false claim may have already changed the rejection path.

\begin{example}\label{ex:placed-only}
Consider two one-seat schools $a$ and $b$ and three students with preferences
\[
 P_1:a\succ b\succ\emptyseto,
 \qquad
 P_2:a\succ\emptyseto,
 \qquad
 P_3:b\succ a\succ\emptyseto.
\]
At school $a$, baseline priority is $1\mathrel{\pi_a}2\mathrel{\pi_a}3$, and at school $b$ it is $1\mathrel{\pi_b}3\mathrel{\pi_b}2$.
Student $1$ has no positive report at either school. At school $a$, student $2$ reports tier $1$ even though $\theta_{2a}=0$, while student $3$ reports tier $2$ and is genuinely eligible for it, $\theta_{3a}=2$.

Under the placed-only rule, students $1$ and $2$ first apply to $a$, while student $3$ applies to $b$. School $a$ provisionally accepts student $2$ at tier $1$ and rejects student $1$. Student $1$ applies to $b$ and displaces student $3$, who then applies to $a$ and displaces student $2$ with her tier-$2$ claim. The provisional outcome is $3\mapsto a$, $1\mapsto b$, and $2\mapsto\emptyseto$.
Only student $3$ makes a positive claim at her provisional assignment, and it is valid. The rule therefore accepts this matching.

Under full verification, student $2$'s false claim gives her tier $0$ at $a$. School $a$ keeps student $1$ in the first round, and student $3$ is never displaced from $b$. The correct outcome is $1\mapsto a$, $3\mapsto b$, and $2\mapsto\emptyseto$.
A false claim by an ultimately unmatched student has changed the outcome through a rejection chain. Checking only provisional assignees is therefore insufficient.
\end{example}

The example separates accuracy from the number of checks. A rule may use fewer checks only because it ignores a claim that changed an earlier rejection. Meaningful comparisons of verification burden must therefore be among procedures that still recover the full-verification outcome. SVDA checks currently pivotal claims before making the relevant school choice, whether or not their claimants are eventually assigned.

SVDA is exact for every audit order, but the order affects the number of checks. When outstanding claims are valid, reported-priority auditing minimizes checks among pivotal-only procedures at each school decision. The next section shows that exact-tier verification makes truthful preference and eligibility reporting a weakly dominant strategy, supplying this validity condition along the recommended path.

\section{Incentives and Truthful-Path Verification}\label{sec:incentives}

A student may change both parts of her report, $\widehat t_i=(\widehat P_i,c_i)$. As usual, \emph{strategy-proofness} means that truthful reporting is a weakly dominant strategy \citep{DubinsFreedman1981,Roth1982}. Here a unilateral deviation may change both preferences and eligibility claims. We do not consider coordinated deviations by groups.

\begin{definition}\label{def:strategyproof-augmented}
Let $\phi(\widehat P,c;\theta)$ denote the final matching induced by a mechanism with verification when the objective eligibility profile is $\theta$, and let $\phi_i$ denote student $i$'s assignment. The mechanism is \emph{strategy-proof} if, for every true profile $(P,\theta)$, every student $i$, every report profile $(\widehat P_{-i},c_{-i})$ of the other students, and every alternative report $(\widetilde P_i,\widetilde c_i)$ by student $i$,
\[
\phi_i\bigl((P_i,\widehat P_{-i}),(\theta_i,c_{-i});\theta\bigr)\mathrel{R_i}\phi_i\bigl((\widetilde P_i,\widehat P_{-i}),(\widetilde c_i,c_{-i});\theta\bigr).
\]
Equivalently, $(P_i,\theta_i)$ is a weakly dominant report for student $i$.
\end{definition}

The weak relation $R_i$ allows truth-telling and a deviation to produce the same assignment.

\subsection{Respecting priority improvements}

We use the familiar \emph{respect for improvements} property of student-proposing DA \citep{BalinskiSonmez1999}; see \citet{Kominers2024} for a general treatment. Our weak version allows priorities to remain unchanged because some reports yield the same effective tier as truth-telling; the proof below follows the standard argument and is included to keep the paper self-contained.

\begin{definition}\label{def:priority-improvement}
Fix a student $i$. A priority profile $\succ'$ is a \emph{weak priority improvement for $i$} over $\succ$ if, at every school $s$,
\begin{enumerate}[label=\textup{(\alph*)}]
    \item the relative priority of every pair of students $j,k\neq i$ is the same under $\succ_s$ and $\succ_s'$, and
    \item whenever $i\succ_s j$, we also have $i\succ_s' j$.
\end{enumerate}
Thus $i$ may move up but never down, while every comparison among other students remains fixed.
\end{definition}

\begin{lemma}\label{lem:respect-improvements}
Fix a preference profile $P$. If $\succ'$ is a weak priority improvement for student $i$ over $\succ$, then $\DA_i(P,\succ')\mathrel{R_i}\DA_i(P,\succ)$.
\end{lemma}

\begin{proof}
Let $\mu=\DA(P,\succ)$ and $\mu'=\DA(P,\succ')$. If $\mu(i)=\emptyseto$, the conclusion follows from the individual rationality of DA. Suppose instead that $\mu(i)=s$ for some school $s$ and, toward a contradiction, that $s\mathrel{P_i}\mu'(i)$.

Construct a report $P_i^s$ that makes $s$ student $i$'s only acceptable school, leaving all other reports unchanged.

We claim that $\mu$ remains stable under $((P_i^s,P_{-i}),\succ')$. Student $i$ is assigned to her only acceptable school and cannot block. Suppose another student $j$ blocks with school $t$. If $t$ has an empty seat, the same pair blocks $\mu$ under $(P,\succ)$. If $j$ competes with an assigned student other than $i$, that priority comparison is unchanged, so the same pair blocks $\mu$ under $(P,\succ)$. If $j$ competes with $i$, then $t=s$. Because $i$ only moved up, $j\succ_s' i$ implies $j\succ_s i$, so $(j,s)$ again blocks the original matching. Both cases contradict the stability of $\mu$. Thus $\mu$ is stable under the modified problem.

Because DA is student-optimal among stable matchings and $\mu$ assigns $i$ to her only acceptable school $s$, DA under $((P_i^s,P_{-i}),\succ')$ also assigns her to $s$.

Under $P_i$ and $\succ'$, however, we assumed that DA gives $i$ the worse outcome $\mu'(i)$. Reporting $P_i^s$ would then profitably obtain $s$, contradicting strategy-proofness of DA under fixed priorities $\succ'$. 
\end{proof}

\subsection{Truthful tier reporting improves one's priorities}

\begin{proposition}\label{prop:tier-improvement}
Fix the true eligibility profile $\theta$, the reports $c_{-i}$ of the other students, and an arbitrary tier report $\widetilde c_i$ by student $i$. The fully validated priority profile generated by the truthful tier report $c_i=\theta_i$ is a weak priority improvement for student $i$ over the fully validated priority profile generated by $\widetilde c_i$.
\end{proposition}

\begin{proof}
Fix school $s$. If $\theta_{is}=0$, student $i$ receives tier $0$ both when truthful and after any false positive claim. Truth-telling cannot lower her priority.

If $\theta_{is}>0$, truth-telling gives $a_{is}=\theta_{is}$. Any claim gives either tier $\theta_{is}$, when it matches $\theta_{is}$, or tier $0$. Truth-telling therefore weakly raises $i$'s tier at every school and leaves all comparisons among other students unchanged. It is a weak priority improvement.
\end{proof}

\begin{theorem}\label{thm:strategyproof-augmented}
SVDA is strategy-proof on the augmented report domain of preferences and eligibility tiers. Equivalently, truthful reporting of both components is a weakly dominant strategy. Formally, for every student $i$, true profile $(P,\theta)$, reports $(\widehat P_{-i},c_{-i})$ of the other students, and deviation $(\widetilde P_i,\widetilde c_i)$,
\[
\SVDA_i\bigl((P_i,\widehat P_{-i}),(\theta_i,c_{-i});\theta\bigr)\mathrel{R_i}\SVDA_i\bigl((\widetilde P_i,\widehat P_{-i}),(\widetilde c_i,c_{-i});\theta\bigr).
\]
\end{theorem}

\begin{proof}
Fix student $i$, the true eligibility profile $\theta$, and the other students' reports $(\widehat P_{-i},c_{-i})$. Consider an arbitrary alternative multidimensional report $(\widetilde P_i,\widetilde c_i)$.

First hold $\widetilde c_i$ fixed and report preferences truthfully. Corollary~\ref{cor:pref-strategyproof} gives
\begin{equation}\label{eq:sp-step1}
\SVDA_i\bigl((P_i,\widehat P_{-i}),(\widetilde c_i,c_{-i});\theta\bigr)\mathrel{R_i}\SVDA_i\bigl((\widetilde P_i,\widehat P_{-i}),(\widetilde c_i,c_{-i});\theta\bigr).
\end{equation}

Next hold $P_i$ fixed and replace $\widetilde c_i$ with $\theta_i$. Proposition~\ref{prop:tier-improvement} makes this a weak priority improvement for $i$. By Corollary~\ref{cor:svda-equivalence}, each SVDA outcome equals the corresponding full-verification DA outcome. Lemma~\ref{lem:respect-improvements} therefore gives
\begin{equation}\label{eq:sp-step2}
\SVDA_i\bigl((P_i,\widehat P_{-i}),(\theta_i,c_{-i});\theta\bigr)\mathrel{R_i}\SVDA_i\bigl((P_i,\widehat P_{-i}),(\widetilde c_i,c_{-i});\theta\bigr).
\end{equation}

Transitivity of $R_i$ combines \eqref{eq:sp-step2} and \eqref{eq:sp-step1}.
\end{proof}

\begin{corollary}\label{cor:truthful-audit-efficiency}
Suppose SVDA uses the reported-priority audit rule. Along the truthful-reporting path, every filed positive tier report is valid. Therefore, at every fixed school and candidate set reached by the mechanism, reported-priority auditing minimizes the number of verifications among pivotal-only local procedures and performs at most $q_s$ checks.
\end{corollary}

\begin{proof}
Under truthful reporting, every positive claim satisfies $c_{is}=\theta_{is}>0$. Proposition~\ref{prop:reported-priority-optimal} and Corollary~\ref{cor:capacity-check-bound} therefore apply at each local decision.
\end{proof}

This corollary links incentives to verification. The audit-order result requires valid claims, and truthful reports satisfy this condition. Strategy-proofness makes truth-telling a weakly dominant recommendation without making it the unique weakly dominant report.

\begin{example}
\label{ex:lie-hurts}
Consider one school $a$ with one seat and eligibility tiers $\mathcal L_a=\{0,1,2\}$. There are two students, both of whom prefer $a$ to remaining unmatched, and baseline priority is $1\mathrel{\pi_a}2$.
Student $1$ has no special eligibility, $\theta_{1a}=0$, while student $2$ is genuinely eligible for tier $1$, so $\theta_{2a}=1$.

If student $2$ reports truthfully, $c_{2a}=1$, verification validates the report. Student $2$ receives effective tier $1$, student $1$ remains at tier $0$, and student $2$ obtains the seat.

Suppose instead that student $2$ reports $c_{2a}=2$. Under exact-tier verification the report is invalid, so student $2$ receives effective tier $0$. Both students are then in tier $0$, and student $1$ obtains the seat by baseline priority. The same outcome results if student $2$ suppresses the entitlement by reporting $c_{2a}=0$. Holding student $1$'s report fixed at tier $0$, student $2$ receives $a$ under $c_{2a}=1$ and $\emptyseto$ under either $c_{2a}=2$ or $c_{2a}=0$.
Because $a\,P_2\,\emptyseto$, both misreports are strictly worse than truthful reporting for student $2$.
\end{example}

Example~\ref{ex:lie-hurts} shows that hiding or misstating a true positive tier can be strictly harmful. Strategy-proofness does not make every false report strictly worse: if a tier cannot affect the allocation, a deviation may leave the outcome unchanged.

\begin{corollary}\label{cor:true-tier-da}
If all students report both their preferences and their eligibility tiers truthfully, then $\SVDA(P,\theta;\theta)=\DA(P,\succ^{\theta})$. Thus truthful SVDA implements the student-optimal stable matching under the true ordered eligibility tiers while verifying only reports that become locally relevant along the DA path.
\end{corollary}

\begin{proof}
Under truthful tier reporting, $c=\theta$, so $a(c,\theta)=\theta$. The result follows immediately from Corollary~\ref{cor:svda-equivalence}.
\end{proof}

Truth-telling need not be the unique weakly dominant report. If a student only adds false positive claims where $\theta_{is}=0$, exact-tier verification produces the same effective tiers as truth-telling. Corollary~\ref{cor:svda-equivalence} then gives the same assignment. With no private reporting or verification cost, these reports are also weakly dominant.

\section{Computational Evidence}\label{sec:simulation}

The theoretical results characterize decision-relevant reports and establish exact implementation of the full-verification DA path. The simulations illustrate the quantitative difference between verifying all submitted claims, verifying claims when their applicants first reach a school, and verifying only currently pivotal claims. By Corollary~\ref{cor:svda-equivalence}, SVDA preserves the full-verification path and matching, so the comparisons concern information acquisition rather than allocation, stability, or welfare. We also examine how the audit order affects the number of checks and illustrate the failure of placed-only verification when a false claim changes a rejection chain.

\subsection{Simulation design}\label{sec:simulation-design}
The baseline market has 1,000 students and 50 schools. Every school has 20 seats, so aggregate capacity equals the number of students. Students and schools are independently located on the unit square. Each school receives an independent standard-normal quality draw. Student $i$'s utility from school $s$ is $u_{is}=Q_s-3d_{is}+\varepsilon_{is}$, where $d_{is}$ is the Euclidean distance between their locations and the shocks $\varepsilon_{is}$ are independent type-I extreme-value draws. Each student ranks her ten highest-utility schools as acceptable; all other schools are below the outside option. Continuous shocks make preferences strict. This spatial random-utility structure captures the central roles of school quality and travel distance in observed school preferences \citep{BurgessGreavesVignolesWilson2015,FackGrenetHe2019} and follows the use of computational models to evaluate school-choice operations \citep{Shi2015}.
A common uniform lottery generates a strict baseline order, which is used at every school. Common lotteries are a familiar design in centralized school choice and feature prominently in the analysis of random tie-breaking in New York City \citep{AbdulkadirogluPathakRoth2009}. Each school has the three ordered tiers $\mathcal L_s=\{0,1,2\}$. The 25 percent of students living closest to the school form its positive-eligibility pool. Conditional on belonging to this pool, a random 20 percent receive tier 2 and the remainder receive tier 1; students outside the pool receive tier 0. The resulting school-level shares are therefore approximately 5 percent in tier 2, 20 percent in tier 1, and 75 percent in tier 0. Tier 2 can be interpreted as an additional school-specific entitlement, such as sibling or statutory priority, among otherwise locally eligible students. The design is deliberately stylized rather than calibrated to one city.

The baseline uses truthful tier reports, so $c_{is}=\theta_{is}$ for every submitted positive claim. Students submit claims only at acceptable schools. For unlisted schools, we normalize both $c_{is}$ and $\theta_{is}$ to zero because those student--school pairs can never enter a school choice. Since distance affects both preferences and eligibility, eligible students are more likely to list the schools at which they receive priority. Consequently, positive claims account for about 42 percent of submitted student--school pairs even though only 25 percent of all students lie in a school's positive-eligibility pool.

This data-generating process is a special case of the theoretical model. Eligibility is school specific; every positive report has exactly the two outcomes in Equation~\eqref{eq:effective-tier}; effective tier is ordered before the public baseline lottery as in Equation~\eqref{eq:priority}; schools use responsive choice with hard capacities; and verified results remain available in later DA rounds.

We compare five exact policies. \emph{Blanket verification} checks every submitted positive claim before DA begins. \emph{On-path full verification} checks every positive claim the first time its claimant applies to the relevant school. The latter is also exact: claims at schools a student never reaches cannot affect any choice on her realized DA path. It therefore separates the benefit of postponing verification until a report enters a candidate set from the additional benefit of testing pivotality. For comparison we also report \emph{placed-only verification}, which honors every claim during DA and then checks only the positive claims made by students at their provisional final assignments. This rule is not exact in general (Example~\ref{ex:placed-only}); we include it as a comparison with a procedure that need not recover the fully verified matching.

The remaining three policies run SVDA exactly as in Definition~\ref{def:svda}. At every school decision, previously held students and new applicants form the candidate set; the procedure identifies the currently pivotal unresolved claims using Proposition~\ref{prop:pivotality-general}, verifies one of them, retains the result, and recomputes pivotality. The policies differ only in their fixed public audit orders:
\begin{enumerate}[label=\textup{(\roman*)}]
    \item \emph{reported priority}: higher reported tiers first, with baseline priority breaking ties, as in Definition~\ref{def:reported-priority-audit};
    \item \emph{baseline priority}: baseline priority alone; and
    \item \emph{random order}: an independent strict random order at each school.
\end{enumerate}
All three orders verify only reports that are currently pivotal. The theory implies that they produce the same choices but may require different numbers of checks.

The baseline contains 300 independently generated markets. Within each market, all policies use the same preferences, capacities, tiers, and lottery. Comparative-static cells use independent markets and vary only the named parameter. The complete design contains 2,770 generated markets, excluding the 250 markets of the placed-only experiment in Subsection~\ref{sec:simulation-invalid}. All five policies are evaluated on every market: the three SVDA audit orders generate 8,310 market--audit-order runs, while the blanket and on-path benchmarks are computed once per market and are therefore common to the three SVDA runs of that market. The master seed is 20260828, and all replication-level data accompany the manuscript.

For each market and audit order, the program compares SVDA with DA under the fully validated effective tiers. It asserts equality of applications and tentative holdings in every round; these equalities also imply identical rejections. It separately checks the final assignment. The program terminates with an error if any equality fails, and no such error occurred in the reported experiments. As an independent implementation check, we compared the count-based pivotality test with direct enumeration of all completions in 2,000 randomly generated local problems; all comparisons agreed. Unless stated otherwise, the reported confidence intervals are normal-approximation 95 percent intervals based on independent market-level replications. These computational checks confirm that the code implements the formal results; they are not substitutes for the proofs.

\subsection{Baseline verification burden}\label{sec:simulation-baseline}

Table~\ref{tab:simulation-baseline} separates two sources of savings. Merely waiting until a positive claim reaches a school reduces the mean number of checks from 4,214.0 to 1,576.6, a reduction of 62.6 percent. Applying the pivotality test produces a further substantial reduction. Reported-priority SVDA performs 999.8 checks on average, approximately one check per student. It reduces verification by 76.3 percent relative to blanket verification and by 36.4 percent relative to on-path full verification.

\begin{table}[H]
\centering
\caption{Verification burden in the baseline market}
\label{tab:simulation-baseline}
\small
\begin{tabular}{lrrrr}
\toprule
Policy & Mean checks & 95\% CI & vs. blanket & vs. on-path \\
\midrule
Blanket verification       & 4,214.0 & [4,200.4, 4,227.5] & ---    & --- \\
On-path full verification  & 1,576.6 & [1,561.9, 1,591.3] & 62.6\% & --- \\
SVDA: random order         & 1,067.3 & [1,058.8, 1,075.8] & 74.7\% & 32.2\% \\
SVDA: baseline priority    & 1,020.5 & [1,012.2, 1,028.8] & 75.8\% & 35.1\% \\
SVDA: reported priority    &   999.8 & [  991.5, 1,008.2] & 76.3\% & 36.4\% \\
\midrule
Placed-only (not exact)    &   802.5 & [  799.5,   805.5] & 81.0\% & 48.8\% \\
\bottomrule
\end{tabular}
\begin{minipage}{0.96\textwidth}
\footnotesize\emph{Notes.} Results are means across 300 independently generated markets. A check resolves one student--school report and has unit cost, as in the model. Savings are computed market by market as one minus the ratio of a policy's checks to the relevant benchmark. Every SVDA run reproduces the fully verified DA path. Placed-only verification recovers the correct matching in these truthful markets but can fail when claims are false (Subsection~\ref{sec:simulation-invalid}).
\end{minipage}
\end{table}

The audit order generates an additional, more modest gain. In paired comparisons on the same markets, reported-priority auditing uses 20.7 fewer checks than baseline-priority auditing, with a 95 percent confidence interval of [19.8, 21.5], and 67.5 fewer than random auditing, with an interval of [65.7, 69.3]. These are reductions of 2.0 and 6.3 percent relative to the respective SVDA alternatives. Exactness does not depend on the order; the differences measure only the information saved by resolving the strongest currently consequential claim first.

Placed-only verification uses fewer checks still: 197.3 fewer than reported-priority SVDA on average, with a 95 percent interval of [190.7, 203.9]. In these truthful markets it also happens to produce the correct matching, because every honored claim is valid. This difference compares the verification burden of the two procedures; it is not a lower bound on the number of checks required for exact implementation. Placed-only verification leaves unchecked claims that can affect earlier rejections without being attached to final assignments. Subsection~\ref{sec:simulation-invalid} shows that once such a claim is false, placed-only verification can accept a wrong matching without observing any failed check.

Across the baseline markets, 26.3 percent of submitted positive reports become pivotal at some point under reported-priority SVDA. This share is not the number ultimately checked: after another report is resolved, a previously pivotal claim may cease to be pivotal and remain unresolved. The pattern illustrates why pivotality is a statement about current decision sensitivity, rather than a claim that every pivotal report must eventually be checked. It also illustrates the adaptive logic of Proposition~\ref{prop:reported-priority-optimal}; by Lemma~\ref{lem:persistence}, a claim that is already nonpivotal cannot become pivotal merely because another claim in the same candidate set is resolved.

The capacity bound also appears exactly at its stated scope. Under truthful reports and reported-priority auditing, no local school-choice invocation uses more than the school's 20 seats in checks. Some invocations attain the bound. Cumulative verification at a school may nevertheless exceed 20 because later applications create new candidate sets. The computation therefore supports the local interpretation of Corollary~\ref{cor:capacity-check-bound}; it does not turn that result into a bound on total verification over the full DA path.

\subsection{Comparative statics}\label{sec:simulation-comparative}

Figure~\ref{fig:simulation-comparative} varies one feature at a time and reports the savings of reported-priority SVDA relative to on-path full verification. Raising the weight on common school quality from 0.5 to 2 makes applications more concentrated at attractive schools and increases the savings rate from 23.4 to 55.1 percent. Increasing the school-level positive-eligibility pool from 10 to 50 percent raises the savings rate from 23.2 to 49.3 percent. In both experiments, the number of on-path full-verification checks increases more rapidly than the number of checks made by SVDA.

Capacity affects the numerator and denominator of the savings rate differently. Mean SVDA checks fall from 1,228.4 when aggregate capacity is 80 percent of the number of students to 638.5 when capacity is 120 percent. Relative savings are nonmonotone: 45.7, 36.3, and 39.4 percent at capacity ratios 0.8, 1.0, and 1.2. Finally, holding average school capacity at 20 and list length at ten, the savings rate rises from 30.0 percent with 500 students to 41.2 percent with 2,500 students. These patterns describe the specified data-generating process. They are not comparative-static theorems, and the market-size exercise is not an asymptotic result.

\begin{figure}[t]
\centering
\includegraphics[width=0.96\textwidth]{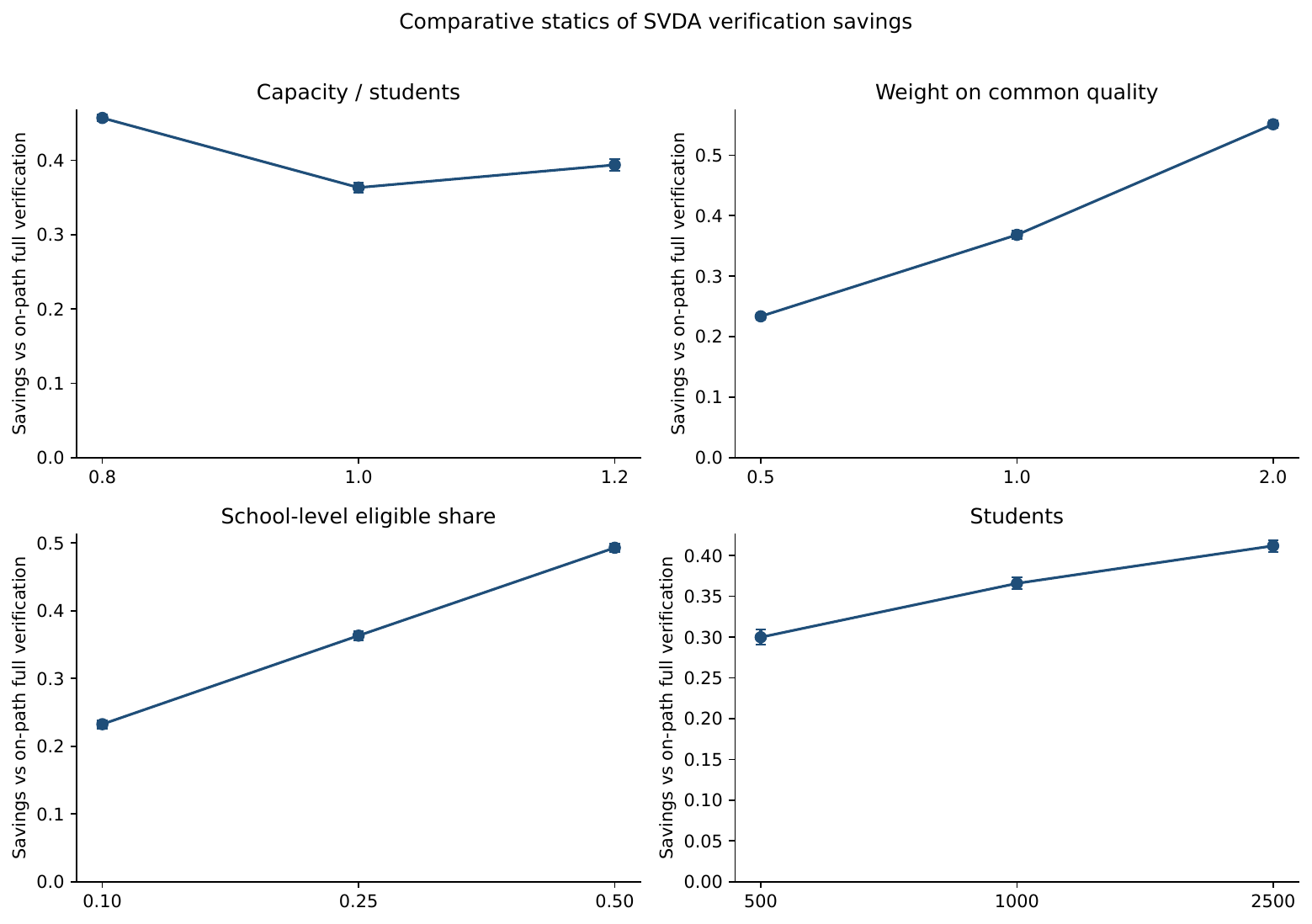}
\caption{Comparative statics of reported-priority SVDA}
\label{fig:simulation-comparative}
\end{figure}

\subsection{Invalid reports and placed-only verification}\label{sec:simulation-invalid}

The baseline deliberately corresponds to the truthful path studied in Corollary~\ref{cor:truthful-audit-efficiency}. To examine robustness outside that condition, we keep the report-generating process fixed and randomly designate 0, 5, 15, or 30 percent of submitted positive reports as invalid by setting the corresponding true tier to zero. The reported tier remains unchanged, and exact-tier verification therefore assigns effective tier zero to the invalid claim, exactly as in Equation~\eqref{eq:effective-tier}.

SVDA remains path-equivalent to fully verified DA in every invalid-report replication, as Corollary~\ref{cor:svda-equivalence} requires. In this separate experiment, with 200 markets at each invalidity level, mean reported-priority checks rise from 1,004.1 with no invalid reports to 1,142.4 when 30 percent are invalid. Savings relative to on-path full verification fall from 36.1 to 28.6 percent. Reported-priority auditing continues to use fewer checks on average than baseline-priority and random auditing in these simulations. This comparison is descriptive only: Proposition~\ref{prop:reported-priority-optimal} assumes that every unresolved claim is valid, so these averages neither extend its domain nor establish dominance for every validation profile.

\begin{figure}[t]
\centering
\includegraphics[width=0.88\textwidth]{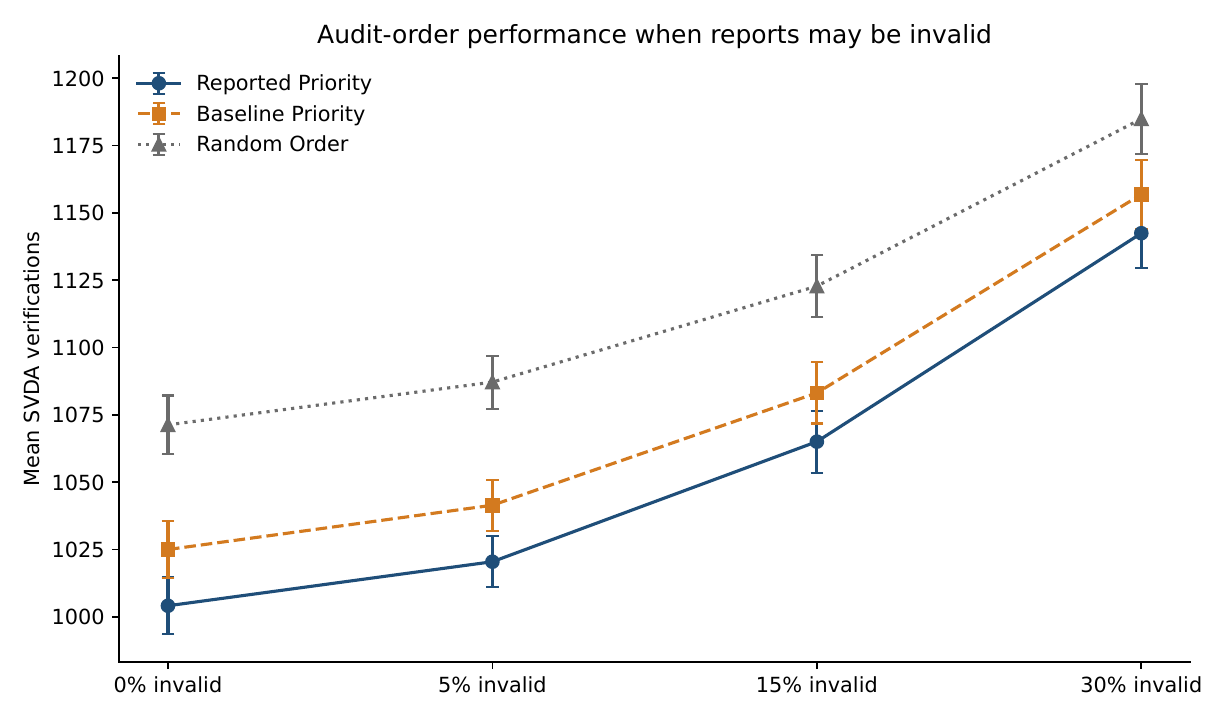}
\caption{Audit-order performance when positive reports may be invalid}
\label{fig:simulation-invalid}
\end{figure}

The final experiment quantifies the rejection-chain concern in Example~\ref{ex:placed-only}. We begin with a truthful market and run reported-priority SVDA to identify positive claims that become pivotal somewhere on the DA path but are not made at the claimant's provisional final assignment. In each of 250 independently generated markets, we sample 50 such claims. We then make each sampled claim false, one at a time, while holding its submitted tier fixed. Because the false claim is not attached to the claimant's provisional assignment and every other report remains truthful, a policy that checks only provisional assignees observes no failed check. We compare the report-honoring provisional assignment accepted by this policy with the assignment under fully validated effective tiers.

Among 12,500 sampled off-assignment pivotal claims, 37 change the final matching, a rate of 0.296 percent with a market-clustered 95 percent confidence interval of [0.202, 0.390] percent. At least one harmful claim appears among the 50 sampled candidates in 14.0 percent of markets, with a 95 percent interval of [9.7, 18.3] percent. Conditional on changing the outcome, a false claim changes the assignments of 4.2 students on average. The experiment is intentionally conditional on selecting an off-assignment claim that was pivotal on the truthful path. It does not estimate the population frequency of false reporting, the unconditional failure probability of placed-only verification, or the welfare loss from a failure. It isolates the mechanism identified by the example: a claim can be absent from the final assignment and still have changed the rejection chain.

\begin{figure}[t]
\centering
\includegraphics[width=0.82\textwidth]{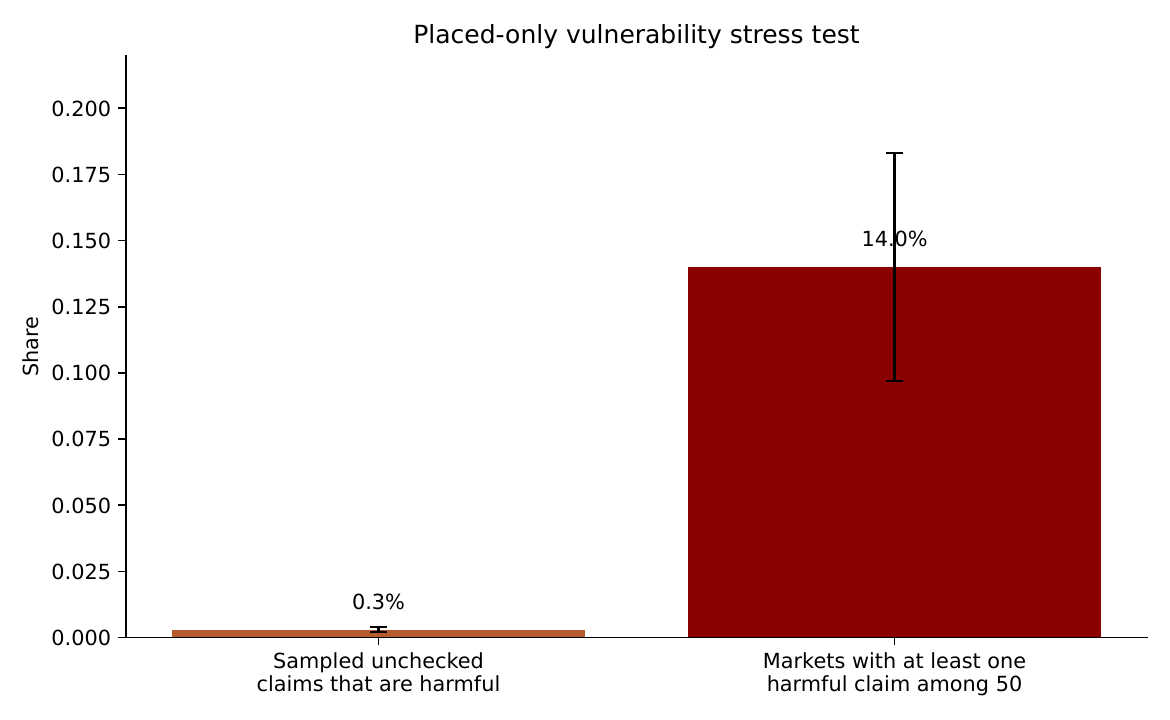}
\caption{Placed-only vulnerability to off-assignment pivotal claims}
\label{fig:simulation-placed-only}
\end{figure}

\FloatBarrier

\subsection{Interpretation}\label{sec:simulation-interpretation}

The simulations illustrate three distinctions in the theoretical analysis. First, postponing verification until a claim reaches a school and conditioning verification on pivotality are separate sources of information savings: in the baseline, the first removes 62.6 percent of blanket checks, and the second removes more than one-third of the remaining checks. Second, the audit order changes the verification burden while leaving the DA path unchanged. Reported-priority auditing uses fewer checks on average in the reported experiments, but the ordering theorem applies only to a fixed school decision with valid unresolved claims. Third, placed-only verification can overlook a false claim that changes a rejection chain. The quantitative magnitudes describe the specified simulated markets; the exactness guarantee follows from the preceding theoretical results.

\section{Discussion and conclusion}\label{sec:discussion}

\subsection{Scope of the analysis}

The results distinguish exact implementation from the number of verifications. Local selective verification recovers the full-verification choice for every pattern of claim validity, and its use at each DA history preserves the entire matching path. The ordering result has a narrower scope. It fixes the school, candidate set, and information state, assumes that all unresolved claims are valid, and compares procedures that query only currently pivotal claims. Likewise, the capacity bound concerns one local invocation. Later applications may change the candidate set and create further verification requirements. The analysis does not characterize the globally minimal number of verifications.

The verification environment treats claims as school specific. This accommodates such facts as a sibling relationship at a particular school or residence within a school-specific zone, but it does not model the joint information conveyed when a single document establishes eligibility at several schools. Shared facts and logically linked verification outcomes would require an explicit account of how a check changes information across school decisions. The local results provide a benchmark for that question without establishing a comparison across entire DA paths.

The validation and priority rules also delimit the conclusions. Under exact-tier verification, a correct positive claim receives its reported tier and an incorrect claim receives tier $0$. Deterministic fallback entitlements, reserves, and seat-specific priorities require corresponding changes to the choice problem \citep{KominersSonmez2016,DurKominersPathakSonmez2018,SonmezYenmez2022}. Heterogeneous or bundled verification costs, noisy evidence, and private documentation costs raise further questions about information acquisition and reporting incentives \citep{BenPorathDekelLipman2014,MylovanovZapechelnyuk2017,BallKattwinkel2025}. These issues lie outside the baseline analysis.

The incentive conclusion is a weak-dominance statement. Truthful reports make all positive claims valid and hence satisfy the hypothesis of the local ordering result, but the model does not select truth-telling uniquely. As Section~\ref{sec:incentives} explains, some false reports induce the same effective priorities and assignments as truthful reports. Thus the verification bound applies on the truthful-reporting path; it is not a conclusion about every profile of weakly dominant reports.

\subsection{Conclusion}

Priority verification in matching is an information problem whose relevant object is the school's current choice. Under ordered eligibility tiers, the paper characterizes precisely when an unresolved claim can change that choice. A claim is pivotal exactly when its claimant can cross the admission cutoff for some common resolution of the other claims. Resolving pivotal claims until the choice becomes invariant implements the full-verification choice, and applying this procedure throughout DA preserves every application, tentative holding, rejection, and final assignment.

When unresolved claims are valid, reported-priority auditing minimizes checks among pivotal-only procedures at each fixed school decision and uses at most the school's capacity in checks. Exact-tier verification also makes truthful preference and eligibility reporting weakly dominant. These results connect the information required by the allocation rule with the incentives created by its validation rule, while retaining the distinction between unconditional exactness and conditional local verification efficiency.

The simulations illustrate the quantitative relevance of this distinction. In the baseline environment, reported-priority SVDA uses approximately one verification per student and 36.4 percent fewer checks than full verification of claims when they first reach a school. The accompanying exercises illustrate variation in the verification burden and the rejection-chain errors that can escape checks confined to provisional assignees. Their role is to complement the theoretical analysis of exact implementation with numerical evidence from the specified synthetic markets.

\section*{Statements and declarations}

\paragraph{Data availability.}
All computational results use synthetic markets generated by the Python code in the accompanying replication package. The package contains the simulation and testing code, replication-level outputs, summary statistics, experiment metadata, software requirements, and instructions for reproducing Table~\ref{tab:simulation-baseline} and Figures~\ref{fig:simulation-comparative}--\ref{fig:simulation-placed-only}. No human-subject or proprietary data are used.

\paragraph{Use of generative artificial intelligence.}
During preparation of this manuscript, ChatGPT (OpenAI) assisted with organization, exposition, and review of the scope of the theoretical claims. Responsibility for the final text and its arguments rests with the authors.

\begin{singlespace}
\interlinepenalty=10000
\bibliographystyle{plainnat}
\bibliography{references}
\end{singlespace}

\end{document}